\documentclass[10pt,a4paper]{article}
\usepackage[a4paper, total={6in, 8in}]{geometry}
\makeatother

\makeatother
\usepackage{bbm}
\usepackage{tabu}
\usepackage{amsmath}
\usepackage{amsthm}
\usepackage{breqn}
\usepackage{verbatim}
\usepackage{graphicx}
\usepackage{caption}
\usepackage{subcaption}
\usepackage{esvect}
\usepackage{cite}
\usepackage{amssymb}
\usepackage{footnote}
\usepackage{color}
\usepackage{url}
\usepackage{xr}
\usepackage[makeroom]{cancel}
\usepackage{microtype}
\theoremstyle{plain}
\newtheorem{theorem}{Theorem}[section]
\newtheorem{corollary}[theorem]{Corollary}
\newtheorem{definition}[theorem]{Definition}

\newtheorem{proposition}[theorem]{Proposition}
\newtheorem{remark}{Remark}
\newtheorem*{problem statement}{Problem Statement}
\newtheorem{feasibility problem}{Problem}

\newtheorem{assumption}{Assumption}
\definecolor{myred}{RGB}{108,0,0}
\graphicspath{{./Figures/}}
\newcommand\blfootnote[1]{%
	\begingroup
	\renewcommand\thefootnote{}\footnote{#1}%
	\addtocounter{footnote}{-1}%
	\endgroup
}
\DeclareMathOperator*{\argmin}{arg\,min}
\newcommand{\real}{\mathbb{R}}
\newcommand{\reg}{\mathcal{R}}

\newcommand{\init}{\mathcal{I}}

\newcommand{\algeb}{\mathcal{F}}
\newcommand{\pr}{\mathrm{Pr}}
\newcommand{\borel}{\mathcal{B}}
\newcommand{\normal}{\mathcal{N}}
\newcommand{\expec}{\mathrm{E}}
\newcommand{\cov}{\mathrm{Cov}}

\newcommand{\sys}{\mathcal{S}}
\newcommand{\proj}{\textbf{proj}}
\newcommand{\uns}{\mathrm{x}_{\mathrm{uns}}}
\newcommand{\imc}{\mathrm{imc}}

\newcommand{\Y}{\mathrm{Y}}

\newcommand{\adv}{\boldsymbol{\pi}}
\begin{document}
	\title{Abstracting the Sampling Behaviour of Stochastic Linear Periodic Event-Triggered Control Systems}
	\author{Giannis Delimpaltadakis, Luca Laurenti, and Manuel Mazo Jr.}
	\date{}
	\maketitle
	\begin{abstract}
		\blfootnote{The authors are with the Delft University of Technology, The Netherlands. Emails:\{i.delimpaltadakis, l.laurenti, m.mazo\}@tudelft.nl. This work is partially supported by the ERC Starting Grant SENTIENT (755953).}
		Recently, there have been efforts towards understanding the sampling behaviour of event-triggered control (ETC), for obtaining metrics on its sampling performance and predicting its sampling patterns. Finite-state abstractions, capturing the sampling behaviour of ETC systems, have proven promising in this respect. So far, such abstractions have been constructed for non-stochastic systems. Here, inspired by this framework, we abstract the sampling behaviour of stochastic narrow-sense linear periodic ETC (PETC) systems via Interval Markov Chains (IMCs). Particularly, we define functions over sequences of state-measurements and interevent times that can be expressed as discounted cumulative sums of rewards, and compute bounds on their expected values by constructing appropriate IMCs and equipping them with suitable rewards. Finally, we argue that our results are extendable to more general forms of functions, thus providing a generic framework to define and study various ETC sampling indicators.
	\end{abstract}
	\section{Introduction}
	Event-Triggered Control (ETC), has been thoroughly studied in the past two decades \cite{astrom2002comparison,tabuada2007etc,girard2015dynamicetc,dynamic_stochastic_etc,stochastic_etc_delay, heemels2012periodic,on_etc_stochastic_tac2020, postoyan_interevent, tallapragada, arman_formal_etc,delimpaltadakis_cdc_2020, delimpaltadakis2020traffic, gabriel_hscc}. 
	The vast majority of ETC research (e.g. \cite{astrom2002comparison,tabuada2007etc,heemels2012periodic,girard2015dynamicetc,dynamic_stochastic_etc,stochastic_etc_delay,on_etc_stochastic_tac2020}) has focused on improving triggering conditions and extending them to wider classes of systems. However, little attention has been paid to evaluating and predicting ETC's sampling behaviour. Such questions are of paramount importance, as answering them would enable: a) deriving performance metrics and evaluating a given ETC design, and b) scheduling traffic in networks of ETC loops. 
	
	Towards reasoning about ETC's sampling behaviour, one approach is based on analytic techniques (\hspace{.1mm}\cite{postoyan_interevent} and \cite{tallapragada}). In \cite{postoyan_interevent} it is shown that, for small-enough triggering-condition parameters, \emph{interevent times} (intervals between consecutive events) converge to certain values or periodic patterns. In \cite{tallapragada}, the nonlinear map describing the evolution of interevent times is studied. Despite the interesting results in \cite{postoyan_interevent} and \cite{tallapragada}, there are certain drawbacks: a) they consider only 2-d linear systems, b) they highly depend on which specific triggering condition is studied, and c) they cannot be easily employed to compute ETC performance metrics in a tractable way.
	
	A different approach is based on \emph{abstractions} \cite{arman_formal_etc,delimpaltadakis_cdc_2020,delimpaltadakis2020traffic,gabriel_hscc}, where a given ETC system is abstracted by a finite-state transition system. 
	The set of the abstraction's output sequences contains all sampling patterns that can be exhibited by the ETC system. 
	Compared to \cite{postoyan_interevent} and \cite{tallapragada}, abstraction-based approaches do not focus on systems of specific dimensions and are not dependent on the given triggering condition (except that some steps in the abstraction's construction might vary). But more importantly, the existing algorithms for computations over finite transition systems enable tractable ways of predicting sampling patterns and computing performance metrics. For instance, as argued in \cite{arman_formal_etc,delimpaltadakis_cdc_2020,delimpaltadakis2020traffic}, such abstractions can be used for scheduling traffic in networks of ETC loops, while, in \cite{gabriel_hscc}, abstractions were used to compute the minimum average interevent time of PETC systems. 
	
	So far, such abstractions have been developed only for non-stochastic systems: \cite{arman_formal_etc} and \cite{gabriel_hscc} considered LTI ETC systems, whereas \cite{delimpaltadakis_cdc_2020} addressed homogeneous systems and \cite{delimpaltadakis2020traffic} extended previous work to general nonlinear systems with bounded disturbances. However, for perturbed systems, the abstractions were more conservative, due to the worst-case scenario approach that was followed. Here, we abstract the sampling behaviour of \emph{stochastic} ETC systems. Compared to non-stochastic, and especially unperturbed, stochastic systems are more realistic. Moreover, regarding disturbances, the probabilistic setting is less conservative,
	since it gives probabilistic assurances, according to the disturbances' probability distributions, instead of providing definitive answers on verification questions, which are bound by the worst case.
	
	We consider stochastic narrow-sense linear PETC (periodic ETC; periodic monitoring of the triggering condition) systems and define their sampling behaviour to be the set of all possible sequences of interevent times and state-measurements along with its induced probability measure. Reasoning about the sampling behaviour can be realized via defining functions over such sequences and computing their expectations. Thus, the problem statement of this work is to compute bounds on such expectations. For clarity, we focus on functions that are expressed as discounted cumulative sums of 
	rewards. Such functions can describe various sampling performance indicators, such as the discounted sum of interevent times, quantifying how frequently the system samples. To compute bounds on such expectations, we construct IMCs (interval markov chains; markov chains with interval transition probabilities, see \cite{givan_bmdps}) capturing PETC's sampling behaviour, define appropriate IMC state-dependent rewards, and employ the algorithms of \cite{givan_bmdps} to compute their expected discounted cumulative sums, serving as the bounds we are looking for. To compute the probability intervals, we study the joint probabilities of transitioning from one region of the state-space to another with the interevent time obtaining a specific value. We show that they can be reformulated as optimization problems of integrals of Gaussians evaluated over polytopes, with their mean varying in different polytopes; such problems have been effectively solved in \cite{luca_tac_2021}. 
	Finally, we argue that our framework is extendable to more general functions, like expected or total rewards, $\omega-$regular properties, etc, which can describe a whole range of sampling-behaviour properties.  
	
	Let us summarize this work's contributions. It is the first one to abstract stochastic ETC sampling. Compared to \cite{arman_formal_etc,delimpaltadakis_cdc_2020,delimpaltadakis2020traffic,gabriel_hscc}, it uses a completely different abstraction framework (IMCs), owing to the need for different mathematical tools to analyze stochastic systems. The probabilities of going from one region to another with the interevent time obtaining a specific value, which can be thought of as reachability analysis for stochastic PETC systems, are investigated here for the first time. In contrast to \cite{arman_formal_etc,delimpaltadakis_cdc_2020,delimpaltadakis2020traffic} which are written in the context of ETC traffic scheduling and to \cite{gabriel_hscc} which studies the minimum average inter-event time, the formalism adopted here, introducing functions of sampling sequences and their expectations, is more generic and can be employed to explore a wide range of sampling behaviour properties. Finally, compared to the literature on IMC-abstractions of stochastic systems (e.g. \cite{luca_tac_2021,lahijanian2015dt_imcs,coogan2020}), it is the first one to employ IMCs for computing bounds on quantitative measures over trajectories, such as cumulative rewards (see Remark \ref{remark_imc_reward_contrib}). In fact, it is shown that these bounds are valid for all trajectories, even though the system's state-space is unbounded, in contrast to \cite{luca_tac_2021,lahijanian2015dt_imcs,coogan2020} which consider bounded domains. 

	\section{Preliminaries}
	\subsection{Notation}
	The symbol $\mathbb{N}_{[0,s]}$ denotes the set of natural numbers up to and including $s$. The $n$-dimensional identity matrix is denoted by $I_n$. For any set $S$, denote its Borel algebra by $\mathcal{B}(S)$. For a set $S\subseteq\real^n$, denote $\overline{S} = \real^n\setminus S$. Given a matrix $T\in\real^{m\times n}$, denote $T\cdot S:=\{T x\in\real^m:x\in S\}$. Denote by $S^k$ the $k$-times Cartesian product $S=S\times\dots\times S$. Given $x\in\real^n$, denote both the $k$-times Cartesian product $\{x\}\times\dots\times\{x\}$ and the $kn$-dimensional vector $\begin{bmatrix}
		x^\top &\dots &x^\top
	\end{bmatrix}^\top$ by $\{x\}^k$. Given sets $Q_1,Q_2$ and $Q=Q_1\times Q_2$, for any $q=(q_1,q_2)\in Q$ denote $\proj_{Q_1}(q)=q_1$ and $\proj_{Q_2}(q)=q_2$. We use the term `path' or `sequence' interchangeably. Given a path $\omega = q_0,q_1,q_2,\dots$, denote $\omega(i)=q_i$. Given a finite path $\omega = q_0,\dots,q_N$, denote $\omega(-1)=q_N$. Finally, $\normal(\mu,\Sigma)$ denotes a Gaussian distribution with mean $\mu$ and covariance matrix $\Sigma$.
	
	\subsection{Interval Markov Chains}
	Interval Markov Chains (IMCs) are finite Markov models, extending discrete-time Markov chains by including uncertainty intervals on transition probabilities. They are often employed as \textit{abstractions} of continuous-space stochastic systems and used for verification (e.g. \cite{lahijanian2015dt_imcs,coogan2020,luca_tac_2021}).
	\begin{definition}[Interval Markov Chain (IMC)]
		An IMC is a tuple $\sys_\imc=\{Q, p_{0,\imc}, \check{P} ,\hat{P}\}$, where: $Q$ is a finite set of states, $p_{0,\imc}:Q\to[0,1]$ is a probability distribution on initial conditions, and $\check{P},\hat{P}:Q\times Q\to[0,1]$ are functions, with $\check{P}(q,q')$ and $\hat{P}(q,q')$ representing lower and upper bounds on the probability of transitioning from state $q$ to $q'$, respectively.
	\end{definition} 
	For all $q,q'\in Q$, we have $\check{P}(q,q')\leq\hat{P}(q,q')$ and $\sum_{q'\in Q}\check{P}(q,q')\leq1\leq\sum_{q'\in Q}\hat{P}(q,q')$. A path of an IMC is a sequence of states $\omega = q_0, q_1, q_2,\dots$, with $q_i\in Q$. Denote the set of the IMC's finite paths by $Paths^{fin}(\sys_{\imc})$. 
	By $p_{0,\imc}(q_0)$ we denote the probability that a path's initial condition is $q_0$. Given a state $q\in Q$, a probability distribution $p_{q}:Q\to[0,1]$ is called \textit{feasible} if $\check{P}(q,q')\leq p_{q}(q')\leq\hat{P}(q,q')$ for all $q'\in Q$. Given $q\in Q$, its set of feasible distributions is denoted by $\Gamma_q$. We denote by $\mathcal{D}(Q)=\{p_q:p_q\in\Gamma_q, q\in Q\}$ the set of all feasible distributions for all states.
	\begin{definition}[Adversary]
		Given an IMC $\sys_{\imc}$, an adversary is a function $\adv:Paths^{fin}(\sys_\imc)\to\mathcal{D}(Q)$, such that $\adv(\omega)\in\Gamma_{\omega(-1)}$, i.e. given a finite path it returns a feasible distribution w.r.t. the path's last element.
	\end{definition}
	The set of all adversaries is denoted by $\Pi$. Given 
	a $\adv\in\Pi$, an IMC path evolves as follows: $\omega(0)$ is sampled according to $p_{0,\imc}$, and at any time-step $i$, given $\omega(i)=q$, $\adv$ chooses a distribution $p_{q}\in\Gamma_q$ from which $\omega(i+1)$ is sampled.
	
	IMCs can be equipped with a reward function $R:Q\to\real_{\geq 0}$. Given an adversary $\adv$, the \textit{expected discounted cumulative sum of rewards} along paths $\expec_\adv(\sum_{i=0}^{N}\gamma^iR(\omega(i)))$, where $\gamma\in[0,1)$ and $N\in\mathbb{N}\cup\{\infty\}$, is well-defined and single-valued; however, due to the existence of infinite adversaries, the IMC produces a whole set of expected values. The bounds of this set ($\sup_{\adv\in\Pi}$ and) $\inf_{\adv\in\Pi}\expec_\adv(\sum_i\gamma^iR(\omega(i)))$ can be computed as shown in \cite{givan_bmdps}. 
	Other tractable computations include total or average rewards \cite{givan_bmdps}, verifying $\omega$-regular properties (\hspace{.1mm}\cite{lahijanian2015dt_imcs,luca_tac_2021,coogan2020}), etc.
	
	\subsection{Stochastic Linear PETC Systems}
	Consider the following stochastic linear control system:
	\begin{equation*}
		d\zeta(t) = A\zeta(t)dt + B\upsilon(t)dt +B_wdW(t),
	\end{equation*}
	where: $A\in\real^{n\times n}$, $B\in\real^{n\times n_\upsilon}$, $B_w\in\real^{n\times n_w}$, $\zeta(t)$ is the 
	state of the system, $\upsilon(t)$ is the control input and $W(t)$ is an $n_w$-dimensional Wiener process on a complete filtered probability space $(\Omega,\mathcal{F}, \{\algeb_t\}_{t\geq0}, \pr)$. $\Omega$ denotes the sample space, $\algeb$ a $\sigma$-algebra, $\{\algeb_t\}_{t\geq0}$ the natural filtration and $\pr$ the probability measure. When the initial condition of the above stochastic differential equation is known, say $x\in\real^n$, we denote the process solving it by $\zeta(t;x)$. Initial conditions are sampled from a probability distribution $p_0:\real^n\to[0,1]$.
	
	In typical state-feedback sample-and-hold control, like ETC, the control input is held constant between consecutive \textit{event time-instants} $t_i$ and $t_{i+1}$:
	\begin{equation}\label{snh}
		d\zeta(t) = A\zeta(t)dt + BK\zeta(t_i)dt +B_wdW(t), \quad t\in[t_i,t_{i+1}),
	\end{equation}
	where $K\in\real^{n_\upsilon\times n}$ is the feedback gain matrix. According to PETC, event time-instants $t_i$ are determined as follows:
	\begin{equation}\label{trig_cond}
		t_{i+1} = t_i + \inf\Big\{t\in\init:\text{} \phi\Big(\zeta(t;\zeta(t_i)),\zeta(t_i)\Big)> 0\Big\}
	\end{equation}
	where $t_0=0$, $\init = \{h,2h,\dots,k_\max h\}$, $h>0$ is a predefined \textit{sampling period}, $k_{\max}>0$, $\phi(\cdot,\cdot)$ is called \textit{triggering function}, $\eqref{trig_cond}$ is called \textit{triggering condition} and $t_{i+1}-t_i$ is called \textit{interevent time}. The state is monitored periodically with period $h$, and when the triggering function is detected positive, then an event time-instant $t_{i+1}$ is defined, and the state measurements $\zeta(t_{i+1})$ are communicated to the controller, which updates the control action to $K\zeta(t_{i+1})$. There is a forced upper-bound $k_\max h$ on interevent times, to prevent the system from operating in an open-loop manner indefinitely. We call the combination \eqref{snh}-\eqref{trig_cond} \textit{(stochastic) PETC system}.
	
	Interevent times are a stochastic process depending on the previously sampled state. Thus, for interevent times, we adopt the following notation: $\tau(x) :=\inf\{t\in\init:\text{} \phi(\zeta(t;x),x)> 0\}$, where $x\in\real^n$ is the previously sampled state.
	\begin{remark}
		Since \eqref{snh} is time-homogeneous, reasoning in the time interval $[t_i,t_{i+1})$ is equivalent to reasoning in $[0,t_{i+1}-t_i)$. In this work, we normally use the latter convention.
	\end{remark}
	\begin{assumption}\label{assum1}
		We assume the following:
		\begin{enumerate}
			\item The matrix pair $(A,B_w)$ is controllable.\label{assum1_controllability}
			\item 
			$\phi(\zeta(t;x),x) = |\zeta(t;x)-x|_\infty-\epsilon$, where $\epsilon>0$ is a predefined constant.\label{assum1_trig_fun}
			\item The sampling period $h=1$ (for ease of presentation).\label{assum1_period}
		\end{enumerate}
	\end{assumption}
	Item \ref{assum1_controllability} ensures that $\zeta(t)$ is a non-degenerate Gaussian random variable (see \cite{luca_tac_2021}). 
	Regarding item \ref{assum1_trig_fun}, $\phi$ is the Lebesgue-sampling function \cite{astrom2002comparison} with an $\infty$-norm instead of a $2$-norm. We restrict ourselves to this case for clarity, but our results are extendable to more general functions. 
	\begin{remark}\label{remark_lebesgue_proof}
		Modifying the proof of {\cite[Theorem 1]{on_etc_stochastic_tac2020}}, it can be proven that the triggering function from Assumption \ref{assum1} guarantees mean-square practical stability for PETC system \eqref{snh}-\eqref{trig_cond}, under mild assumptions. The proof is omitted due to space limitations.
	\end{remark}
	
	\section{Problem Formulation}
	The stochastic PETC system \eqref{snh}-\eqref{trig_cond} can exhibit different sequences of communicated measurements and interevent times $(\zeta(t_0),0),(\zeta(t_1;\zeta(t_0)),t_1),(\zeta(t_2-t_1;\zeta(t_1)),t_2-t_1),\dots$, depending on initial conditions and random events. We denote the set of all possible such sequences of infinite length by:
	\begin{align*}
		\Y=\{(x_0,0),(x_1,s_1),(x_2,s_2),\dots| \text{ }x_i\in\real^n, s_i\in\mathbb{N}_{[0,k_\max]}\}
	\end{align*} 
	We call $\Y$ the \textit{sampling behaviours} of the PETC system. There is a well-defined probability measure $\pr_Y$ over $\borel(\Y)$ (according to \cite{ionescu}), induced by $\pr$ and $p_0$ as follows: 
	\begin{align}
		&\pr_\Y(\omega(0)\in (X_0,s)) =\left\{\begin{aligned}&\int_{X_0}p_0(x)dx \text{, if }s=0\\&0, \text{ otherwise}\end{aligned}\right.\label{prob_init_sys}\\
		&\pr_\Y(\omega(i+1)\in (X_{i+1},s_{i+1})|\omega(i)=(x_i,s_i)) =\nonumber\\ &\pr(\zeta(s_{i+1};x_i)\in X_{i+1}, \tau(x_i)=s_{i+1})
	\end{align}
	where $\omega\in\Y$, 
	$s,s_i,s_{i+1}\in\mathbb{N}_{[0,k_\max]}$, $x_i\in\real^n$, $X_0,X_{i+1}\subseteq\real^n$ and we use $(X_0,s)$ (or $(X_{i+1},s_{i+1})$) to denote the set $\{(x,s):x\in X_0\}$. Also, let us denote $Q:=\real^n\times\mathbb{N}_{[0,k_\max]}$. 

	Studying the PETC system's sampling behaviour can be formalized by defining functions $f:\Y\to\real$ and computing their expectation $\expec_{\pr_\Y}(f(\omega))$. Here, we focus on functions that can be described as cumulative discounted rewards:
	\begin{equation}\label{cumulative_sum}
		f(\omega) = \sum_{i=0}^\infty\gamma^iR(\omega(i))
	\end{equation}
	where $R:Q\to[0,R_\max]$ is a bounded reward function. With a suitable choice of $R$, $\expec_{\pr_\Y}(f(\omega))$ can describe various indicators on PETC's sampling behaviours and performance:
	
	\textit{Example 1:} Consider the reward $R((x,s))=s$. Then, $f(\omega)$ represents the discounted sum of interevent times over paths. The expectation $\expec_{\pr_\Y}(f(\omega))$ provides a useful metric on the PETC system's sampling performance: the bigger it is, the bigger are expected to be the interevent intervals, which implies that the expected sampling performance is ``better".
	
	\textit{Example 2:} Consider the reward $R((x,s))=\min(\alpha\tfrac{1}{|x|+\varepsilon}+\beta s,R_\max)$, with $\alpha,\beta,\varepsilon>0$, which penalizes paths that overshoot far from the origin or exhibit a high sampling frequency. Again, a bigger $\expec_{\pr_\Y}(f(\omega))$ implies better performance.
	
	Unfortunately, exactly computing expectations $\expec_{\pr_\Y}(f(\omega))$ is infeasible: among others, how does one obtain $\pr_\Y$ and integrate over an uncountable set of paths $\Y$? This motivates the following problem statement:
	\begin{problem statement}
		Consider the PETC system \eqref{snh}-\eqref{trig_cond}, its sampling behaviours $\Y$ along with $\pr_\Y$, and let Assumption \ref{assum1} hold. 
		Given a reward $R:Q\to[0,R_\max]$ and its discounted cumulative sum $f(\omega)=\sum_{i}\gamma^iR(\omega(i))$, with $\gamma\in[0,1)$, compute (non-trivial) lower/upper bounds on $\expec_{\pr_\Y}(f(\omega))$.
	\end{problem statement}
	In what follows, the problem is addressed by abstracting the sampling behaviour $Y$ and $\pr_\Y$ by an appropriate IMC $\sys_\imc$, defining suitable reward functions $\underline{R},\overline{R}$ over its states, and calculating $\inf_{\adv\in\Pi}\expec_{\adv}(\sum_{i}\gamma^i\underline{R}(\omega(i)))$ and $\sup_{\adv\in\Pi}\expec_{\adv}(\sum_{i}\gamma^i\overline{R}(\omega(i)))$. 
	\begin{remark}
		Our results are extendable to more general functions $f(\omega)$, such as expected or total rewards, $\omega$-regular properties, etc. (by modifying our proofs in accordance to e.g. \cite{givan_bmdps,luca_tac_2021}). Thus, leveraging our results, we could compute bounds on: expectation and variance of interevent times, probability of certain sampling patterns arising, etc.
	\end{remark}
	\section{Abstracting the Sampling Behaviour via IMCs} 
	Typically, to abstract a system with a continuous state-space via a finite-state IMC, the following steps are followed (e.g. \cite{lahijanian2015dt_imcs,coogan2020,luca_tac_2021}): 1) the state-space is partitioned into a finite number of sets, each of which is represented by a state of the IMC, 2) if the state-space is unbounded, then one set of the partition is unbounded as well, and its corresponding IMC-state is made absorbing, and 3) given two states $q_1,q_2$ of the IMC, with $q_1$ non-absorbing, the transition probability intervals from $q_1$ to $q_2$ are computed such that they bound the probability $\pr(\omega(i+1)\in q_2|\omega(i)=x_1)$ for all $x_1\in q_1$, where $\omega$ refers to paths of the original system.
	
	Here, we employ the same ideas. Consider a compact polytope $X\subset\real^n$, $m$ convex polytopes $\reg_i$ such that $\bigcup\limits_{i=1}^{m}\reg_i = X\subset \real^n$ and the following IMC:
	\begin{equation}\label{our_imc}
		\sys_\imc =(Q_\imc,p_{0,\imc}, \check{P},\hat{P}),
	\end{equation}
	where:
	\begin{itemize}
		\item $Q_\imc = (Q_\reg\times\mathbb{N}_{[0,k_\max]})\cup\{\uns\}$, where $Q_\reg = \{\reg_1,\reg_2,\dots,\reg_m\}$ and $\uns$ is an indicator for all states  of the PETC system belonging in $\overline{X}$.
		\item $p_{0,\imc}:Q_\imc\to[0,1]$ is such that: 
		\begin{equation}\label{prob_init_imc}
			p_{0,\imc}(q) = \left\{\begin{aligned}
				&\int_{\proj_{Q_\reg}(q)}p_0(x)dx, \text{ if }q\neq\uns \text{ and}\\&\qquad\qquad\qquad\qquad\proj_{\mathbb{N}_{[0,k_\max]}}(q)=0\\
				&\int_{\overline{X}}p_0(x)dx, \text{ if }q=\uns\\
				&0,\text{ otherwise}
			\end{aligned}\right.
		\end{equation}
		\item $\check{P}$ and $\hat{P}$ are such that $\forall(\reg,s),(\reg',s')\in Q_\imc\setminus\{\uns\}$: 
		\begin{equation}\label{probability_bounds}
			\begin{aligned}
				&\check{P}\Big((\reg,s),(\reg',s')\Big)\leq \min\limits_{x\in\reg}\pr(\zeta(s';x)\in \reg',\tau(x)=s')\\
				&\hat{P}\Big((\reg,s),(\reg',s')\Big)\geq \max\limits_{x\in\reg}\pr(\zeta(s';x)\in \reg',\tau(x)=s')\\
				&\check{P}\Big((\reg,s),\uns\Big)\leq\\&\qquad\qquad\min\limits_{(x,s')\in\reg\times\mathbb{N}_{[0,k_\max]}}\pr(\zeta(s';x)\in \overline{X},\tau(x)=s')\\
				&\hat{P}\Big((\reg,s),\uns\Big)\geq\\&\qquad\qquad\max\limits_{(x,s')\in\reg\times\mathbb{N}_{[0,k_\max]}}\pr(\zeta(s';x)\in \overline{X},\tau(x)=s')
			\end{aligned}
		\end{equation}
		and for all $q'\in Q_\imc$:
		\begin{equation}\label{unsafe_state_probabilities}
			\check{P}(\uns,q')=\hat{P}(\uns,q') = \left\{\begin{aligned}
				&1,\text{ if }q'=\uns\\ &0,\text{ otherwise}
			\end{aligned}\right.
		\end{equation}
	\end{itemize}
	To show how $\sys_\imc$ abstracts the sampling behaviours $\Y$, let us relate paths of the IMC  to paths in $\Y$. First, consider any path $\omega=(x_0,0),(x_1,s_1),\dots\in\Y$ for which $\not\exists j\geq0$ such that $x_j\in\overline{X}$. Such a path is related to a path $\tilde{\omega}$ in the IMC, which is such that $x_i\in\proj_{Q_\reg}(\tilde{\omega}(i))$ and $s_i = \proj_{\mathbb{N}_{[0,k_\max]}}(\tilde{\omega}(i))$ for all $i$. Next, consider paths $\omega=(x_0,0),(x_1,s_1),\dots\in\Y$ for which $\exists j\geq0$ such that $x_j\in\overline{X}$ and $x_i\notin\overline{X}$ for all $i<j$. These are related to IMC paths $\tilde{\omega}$, which are such that $x_i\in\proj_{Q_\reg}(\tilde{\omega}(i))$ and $s_i = \proj_{\mathbb{N}_{[0,k_\max]}}(\tilde{\omega}(i))$ for all $i<j$, and $\tilde{\omega}(i)=\uns$ for all $i\geq j$. Note that $\uns$ is absorbing, and all paths in $\Y$ that enter $\overline{X}$ (even those that eventually return to $X$) are mapped to IMC paths that enter $\uns$ at the same time and stay there.
	
	\begin{remark}
		$X$ is assumed to be a polytope and $\reg_i$ to be convex polytopes in order to facilitate optimization techniques employed later to determine $\check{P}$ and $\hat{P}$.
	\end{remark}
	The above IMC, with a suitable choice of rewards, can be used to address the problem at hand:
	\begin{theorem}\label{main_theorem}
		Consider a reward function $R:Q\to[0,R_\max]$ and $f(\omega)$ as in \eqref{cumulative_sum}. Consider the IMC $\sys_\imc$ from \eqref{our_imc}.
		Define reward functions $\underline{R},\overline{R}:Q_\imc\to[0,R_\max]$: 
		\begin{equation}\label{underline_R}
			\begin{aligned}
				&\underline{R}(q) = \left\{\begin{aligned}
					&\min\limits_{(x,s)\in q}R((x,s)), \text{ if }q\neq\uns\\&\min\limits_{(x,s)\in Q}R((x,s)), \text{ if }q=\uns
				\end{aligned}\right.\\
				&\overline{R}(q) = \left\{\begin{aligned}
					&\max\limits_{(x,s)\in q}R((x,s)), \text{ if }q\neq\uns\\&\max\limits_{(x,s)\in Q}R((x,s)), \text{ if }q=\uns
				\end{aligned}\right.
			\end{aligned}
		\end{equation}
		Then, the following holds:
		\begin{equation*}
			\inf_{\adv\in\Pi}\hspace{-.5mm}\expec_\adv(\sum_{i=0}^\infty\gamma^i\underline{R}(\tilde{\omega}))\leq\expec_{\pr_\Y}(f(\omega))\leq\sup_{\adv\in\Pi}\hspace{-.5mm}\expec_\adv(\sum_{i=0}^\infty\gamma^i\overline{R}(\tilde{\omega}))
		\end{equation*}
	\end{theorem}
	\begin{proof}
		Below, we treat $\uns$ as the set $\overline{X}\times\mathbb{N}_{[0,k_\max]}$; whenever $y\in \overline{X}\times\mathbb{N}_{[0,k_\max]}$ we write $y\in\uns$. Also, for $s_i,s_{i+1}\in\mathbb{N}_{[0,k_\max]}$, $x_i\in\real^n$, $X_{i+1}\subseteq\real^n$,  denote $T((X_{i+1},s_{i+1})|(x_i,s_i)):=\pr_\Y(\omega(i+1)\in (X_{i+1},s_{i+1})|\omega(i)=(x_i,s_i))$ ($T$ is often called \textit{transition kernel}).
		With abuse of notation, we write $\int_QT(dy'|y)$, for some $y\in Q$, to denote $\sum_{s'\in\mathbb{N}_{[0,k_\max]}}\int_{\real^n}T((dx',s')|y)$.
		
		We focus on the lower bound, as the upper bound's proof follows similarly. It suffices to show that $\exists\adv\in\Pi$ such that:
		\begin{equation}\label{prove_this}
			\expec_\adv(\sum_{i=0}^\infty\gamma^i\underline{R}(\tilde{\omega}))\leq\expec_{\pr_\Y}(f(\omega))
		\end{equation}
		We constrain our search to history-independent (\textit{Markovian}) adversaries, which can be defined as $\adv:Q_\imc\to\mathcal{D}(Q_\imc)$. Given that a Markovian adversary, given a state $q\in Q_\imc$, returns a feasible distribution $\adv(q)=p_q\in\Gamma_q$, by abuse of notation we write $\adv(q,q')=p_q(q')$ for any $q'\in Q_\imc$. 
		
		Define the so-called \textit{value functions}:
		\begin{align*}
			&V(y) = R(y) + \gamma\int_{Q}V(y')T(dy'|y), \text{ }\forall y\in Q\\
			&\underline{V}_\adv(q) = \underline{R}(q) + \gamma\sum_{q'\in Q_\imc}\underline{V}_\adv(q')\adv(q,q'), \text{ }\forall q\in Q_\imc
		\end{align*}
		$V(y)$ is the expected value of $f(\omega)$ given that $\omega(0)=y$: $V(y) = \expec_{\pr_\Y}(f(\omega)|\omega(0)=y)$. Similarly, $\underline{V}_\adv(q) = \expec_{\adv}(\sum_{i}\gamma^i\underline{R}(\tilde{\omega}(i))|\tilde{\omega}(0)=q)$ (see \cite{puterman}). Thus:
		\begin{align*}
			\expec_{\pr_\Y}(f(\omega)) &= \int_QV(y)\pr(\omega(0)\in y)dy\nonumber\\
			\expec_{\adv}(\sum_{i=0}^\infty\gamma^i\underline{R}(\tilde{\omega}(i))) &= \sum_{q\in Q_\imc}\underline{V}_\adv(q)p_{0,\imc}(q)
		\end{align*}
		By incorporating \eqref{prob_init_sys} and \eqref{prob_init_imc} to the above equations, we get:
		\begin{equation}\label{exp_sys}
			\begin{aligned}
				&\expec_{\pr_\Y}(f(\omega)) = \int_{\real^n}V((x,0))p_0(x)dx=\\&=\int_{\overline{X}}V((x,0))p_0(x)dx+\sum_{\reg\in Q_\reg}\int_{\reg}V((x,0))p_0(x)dx\\
			\end{aligned}
		\end{equation}
		\begin{equation}\label{exp_imc}
			\begin{aligned}
				&\expec_{\adv}(\sum_{i=0}^\infty\gamma^i\underline{R}(\tilde{\omega}(i)))=\\
				&=\underline{V}_\adv(\uns)\int_{\overline{X}}p_0(x)dx+\sum_{\reg\in Q_\reg}\underline{V}_\adv((\reg,0))\int_{\reg}p_0(x)dx
			\end{aligned}
		\end{equation}
		Observe that if we prove that there exists a $\adv$ such that the following two conditions hold:
		\begin{align}
			&\underline{V}_\adv(\uns)\leq\inf_{y\in Q}V(y),\label{v_uns_leq}\\
			&\forall q\in Q_\imc\setminus\uns:\quad \underline{V}_\adv(q)\leq\min_{y\in q}V(y)\label{v_q_leq}
		\end{align}
		then from \eqref{exp_sys}-\eqref{exp_imc} we have that \eqref{prove_this} holds and the proof is complete. Consider the following adversary for all $q\in Q_\imc$:
		\begin{equation*}
			\adv(q,q') = \left\{\begin{aligned}
				&\int_{q'}T(dy'|y^\star(q)), \text{ if }q\neq\uns\\
				&1, \text{ if }q=q'=\uns,\\
				&0, \text{ otherwise }
			\end{aligned}\right.
		\end{equation*}
		where $y^\star(q)=\argmin_{y\in q}V(y)$. Indeed $\adv\in\Pi$, since $\check{P}(q,q')\leq\adv(q,q')\leq\hat{P}(q,q')$ and $\sum_{q'\in Q_\imc}\adv(q,q')=1$ for all $q\in Q_\imc$. We will show that $\adv$ satisfies \eqref{v_uns_leq} and \eqref{v_q_leq}, thus completing the proof. We use the fact that $\underline{V}_\adv(q)$ is the fixed-point of the following \textit{value iteration} \cite{givan_bmdps}:
		\begin{equation*}\label{val_iter}
			\begin{aligned}
				&\underline{V}_{\adv,0}(q) = \underline{R}(q),\\
				&\underline{V}_{\adv,i+1}(q) = \underline{R}(q) + \gamma\sum_{q'\in Q_\imc}\underline{V}_{\adv,i}(q')\adv(q,q'), \text{ }\forall q\in Q_\imc
			\end{aligned}
		\end{equation*}
		
		Let us now prove \eqref{v_uns_leq} first, via induction. Observe that:
		\begin{equation*}
			\underline{V}_{\adv,0}(\uns) = \underline{R}(\uns)\leq \inf_{y\in Q}V(y)
		\end{equation*}
		due to \eqref{underline_R} and the fact that $V(y)\geq R(y)$ for all $y\in Q$. Now, if we assume that $\underline{V}_{\adv,i}(\uns) \leq \inf_{y\in Q}V(y)$, we have:
		\begin{align*}
			\underline{V}_{\adv,i+1}(\uns) &= \underline{R}(\uns) + \gamma\sum_{q'\in Q_\imc}\underline{V}_{\adv,i}(q')\adv(\uns,q')\\
			&=\underline{R}(\uns) + \gamma\underline{V}_{\adv,i}(\uns)\\
			&\leq\inf_{y\in Q}R(y) + \gamma \inf_{y\in Q}(V(y))\int_{Q}T(dy'|y_0),\forall y_0\\
			&\leq\inf_{y_0\in Q}\Big(R(y_0) + \gamma \int_{Q}V(y')T(dy'|y_0)\Big) \\&= \inf_{y\in Q}V(y),
		\end{align*}
		where in the second step we used that $\adv(\uns,\uns)=1$ and $\adv(\uns,q')=0$ for any $q'\neq\uns$, in the third step we used that $\int_{Q}T(dy'|y_0)=1$ for any $y_0$, and in the fourth step we used that $\inf_{y\in Q}(V(y))\leq V(y')$ for all $y'\in Q$. Thus, by induction, we have proven \eqref{v_uns_leq}.
		
		Finally, let us prove \eqref{v_q_leq}, again by induction. Again from \eqref{underline_R}, we have that $\underline{V}_{\adv,0}(q) \leq \min_{y\in q}V(y)$  for all $q\in Q_\imc\setminus\uns$. If we assume that $\underline{V}_{\adv,i}(q) \leq \min_{y\in q}V(y)$ for all $q\in Q_\imc\setminus\uns$, then we have for all $q\in Q_\imc\setminus\uns$:
		\begin{align*}
			\underline{V}_{\adv,i+1}(q) &= \underline{R}(q) + \gamma\sum_{q'\in Q_\imc}\underline{V}_{\adv,i}(q')\adv(q,q')\\
			&=\underline{R}(q) + \gamma\sum_{q'\in Q_\imc}\underline{V}_{\adv,i}(q')\int_{q'}T(dy'|y^\star(q))\\
			&\leq\min_{y\in q}R(y) + \\&\qquad+\gamma \sum_{q'\in Q_\imc}\min_{y\in q'}(V(y))\int_{q'}T(dy'|y^\star(q))\\
			&\leq R(y^\star(q)) + \gamma\sum_{q'\in Q_\imc}\int_{q'}V(y')T(dy'|y^\star(q))\\
			&=V(y^\star(q))=\min_{y\in q}V(y)
		\end{align*}
		where in the third step we used that $\underline{V}_{\adv,i}(q) \leq \min_{y\in q}V(y)$ for all $q\in Q_\imc\setminus\uns$ by the induction assumption and $\underline{V}_{\adv,i}(\uns) \leq \inf_{y\in Q}V(y)$ by our earlier proof, in the fourth step we used that $\min_{y\in q'}(V(y))\leq V(y')$ for all $y'\in q'$ and $\min_{y\in q}R(y)\leq R(y^\star(q))$, and in the fifth step we used that $y^\star(q) = \argmin_{y\in q}V(y)$. The proof is complete.
	\end{proof}
	\begin{remark}\label{remark_imc_reward_contrib}
		To the authors' knowledge, this is the first time, in the literature of IMC-abstractions of stochastic systems, that IMCs are employed to compute bounds on quantitative measures over a system's trajectories, such as cumulative rewards. In addition, we highlight that the bound holds even in the case where the system's state-space is unbounded and an unbounded set of it (the set $\overline{X}$) is abstracted by an absorbing state $\uns$. That is, the bound does take into account paths of the system that eventually leave $X$ and certainly contribute to the expectation of the cumulative reward.
	\end{remark}
	\begin{remark}
		The above proof does not provide a recipe to compute the bounds on $\expec_{\pr_\Y}(f(\omega))$, since this would assume knowledge of $y^\star(q)$; it only shows that there exists an adversary $\adv\in\Pi$ such that $\expec_{\adv}(\sum_{i}\gamma^i\underline{R}(\tilde{\omega}(i)))\leq\expec_{\pr_\Y}(f(\omega))$. It has been shown in \cite{givan_bmdps} that the bound ($\sup_{\adv\in\Pi}$ and) $\inf_{\adv\in\Pi}\expec_{\adv}(\sum_{i}\gamma^i\underline{R}(\tilde{\omega}(i)))$ can be obtained via a modified value iteration algorithm, with polynomial time-complexity. However, an important aspect that our proof shows is that value iteration provides a valid bound on  $\expec_{\pr_\Y}(f(\omega))$ in every time step $i$; thus, the algorithm can be terminated in an arbitrary number of steps, still providing sound results.
	\end{remark}
	Thus, given $\sys_{\imc}$, by defining rewards $\underline{R},\overline{R}$ as described in Theorem \ref{main_theorem}, and calculating $\inf_{\adv\in\Pi}\expec_{\adv}(\sum_{i}\gamma^i\underline{R}(\tilde{\omega}(i)))$ and $\sup_{\adv\in\Pi}\expec_{\adv}(\sum_{i}\gamma^i\overline{R}(\tilde{\omega}(i)))$ via the algorithm proposed in \cite{givan_bmdps}, we obtain non-trivial bounds on $\expec_{\pr_\Y}(f(\omega))$. To construct $\sys_\imc$, what remains is to determine $\check{P},\hat{P}$, according to \eqref{probability_bounds}. This is carried out in the next section.
	
	\section{Transition Probability Intervals}
	In this section, we derive $\check{P}$ and $\hat{P}$, according to \eqref{probability_bounds}. In what follows, for $s\in\mathbb{N}_{[0,k_\max]}$, we denote $\zeta(s;x)=\zeta_{s,x}$ and $\tilde{\zeta}_{s,x} = \begin{bmatrix}
		\zeta^\top_{1,x} &\zeta^\top_{2,x} &\dots &\zeta^\top_{s,x}
	\end{bmatrix}^\top$.
	
	Let us investigate equation \eqref{probability_bounds}, which indicates that we are interested in quantities ($\max_{x\in\reg}$ or) $\min_{x\in\reg}\pr(\zeta_{s',x}\in S, \tau(x)=s')$, where $S=\reg'$ or $S=\overline{X}$. The law of conditional probabilities implies: 
	\begin{equation}\label{IV_eq1}
		\begin{aligned}
			\min\limits_{x\in\reg}\pr(\zeta_{s',x}\in S,\tau(x)=s') \geq&\\ \min\limits_{x\in\reg}\pr(\zeta_{s',x}\in S|\tau(x)=s')\cdot\min\limits_{x\in\reg}\pr(\tau(x)=s')&
		\end{aligned}
	\end{equation} 
	and conversely for $\max_{x\in\reg}$. Hence, to determine $\check{P}$ and $\hat{P}$ according to \eqref{probability_bounds}, it suffices to minimize and maximize over $\reg$ the quantities $\pr(\tau(x)=s')$ and $\pr(\zeta_{s',x}\in S|\tau(x)=s')$. 
	For conciseness, we assume that $S$ is a compact polytope, since in the case of the unbounded $S=\overline{X}$ we can write:
	\begin{equation*}
		\pr(\zeta_{s',x}\in \overline{X}|\tau(x)=s') = 1 - \pr(\zeta_{s',x}\in X|\tau(x)=s'),
	\end{equation*}
	and focus on $\pr(\zeta_{s',x}\in X|\tau(x)=s')$ with $X$ being compact. In what follows, it is shown that optimizing over $\reg$ the quantities $\pr(\tau(x)=s')$ and $\pr(\zeta_{s',x}\in S|\tau(x)=s')$, can be reformulated as optimizing the integral of a Gaussian evaluated over a polytope, with its mean varying in another polytope. Such optimization problems have been effectively solved in \cite{luca_tac_2021}. We focus on minimization, as maximization is the same up to a reversal of inequalities.
	
	The following proposition, stemming from Assumption \ref{assum1}, is instrumental in our construction, as it enables computing probabilities of the type $\pr(\tilde{\zeta}_{s,x}\in S)$ as an integral of a Gaussian distribution over $S$:
	\begin{proposition}\label{normal}
		It holds that $\tilde{\zeta}_{s,x}\sim\normal(\mu_{\tilde{\zeta}_{s,x}},\Sigma_{\tilde{\zeta}_{s,x}})$, with $\mu_{\tilde{\zeta}_{s,x}}=\begin{bmatrix}
			\expec(\zeta^\top_{1,x})&\expec(\zeta^\top_{2,x})&\dots&\expec(\zeta^\top_{s,x})
		\end{bmatrix}^\top$,
		\begin{align*}
			&\Sigma_{\tilde{\zeta}_{s,x}}=\begin{bmatrix}
				\cov(1,1) &\cov(1,2)&\dots&\cov(1,s)\\
				\vdots &\vdots &\dots &\vdots\\
				\cov(s,1) &\cov(s,2) &\dots &\cov(s,s)
			\end{bmatrix}
		\end{align*}
		where $\expec(\zeta(t;x)) = [e^{At}(I+A^{-1}BK)-A^{-1}BK]x$,
		\begin{align*}
			&\cov(t_1,t_2) = \int_{0}^{\min(t_1,t_2)}e^{A(t_1-s)}B_wB_w^\top e^{A^\top(t_2-s)}ds
		\end{align*}
	\end{proposition}
	\begin{proof}
		Application of the expectation and covariance operators to the solution of linear SDE \eqref{snh} (see {\cite[pp. 96]{mao_book}}).
	\end{proof}
	\subsection{Probabilities on Interevent Times}
	Here, we focus on the second term of \eqref{IV_eq1}: $\min_{x\in\reg}\pr(\tau(x)=s)$; i.e., the probability that the interevent time is $s$ when starting from $\reg$. Define the set:
	\begin{equation*}
		\Phi(x):=\{y\in\real^n: \phi(y,x)\leq 0\}=\{y\in\real^n:|y-x|_\infty\leq\epsilon\},
	\end{equation*}
	which is such that $\phi(\zeta(t;x),x)>0\iff\zeta(t;x)\notin\Phi(x)$, where $\phi(\cdot)$ is the triggering function. Hence, for any $s<k_\max$, since the probability that the interevent time is $s$ is equal to the probability that the triggering function was negative at times $1,2,\dots,s-1$ and positive at $s$, we can write:\footnote{With a slight abuse of notation, $\Phi^0(x)=\{x\}$.}
	\begin{equation*}
		\begin{aligned}
			\pr(\tau(x)=s) = \pr(\tilde{\zeta}_{s,x}\in\Phi^{s-1}(x)\times\overline{\Phi}(x))=&\\
			\pr(\tilde{\zeta}_{s-1,x}\in\Phi^{s-1}(x))-\pr(\tilde{\zeta}_{s,x}\in\Phi^{s}(x))&
		\end{aligned}
	\end{equation*}
	The above implies:
	\begin{equation}\label{yo}
		\begin{aligned}
			\forall s<k_\max:\quad\min\limits_{x\in\reg}\pr(\tau(x)=s)\geq&\\
			\min\limits_{x\in\reg}\pr(\tilde{\zeta}_{s-1,x}\in\Phi^{s-1}(x))-\max\limits_{x\in\reg}\pr(\tilde{\zeta}_{s,x}\in\Phi^{s}(x))&
		\end{aligned}
	\end{equation}
	For $s=k_\max$, we have $\pr(\tau(x)=k_\max)=\pr(\tilde{\zeta}_{k_\max-1,x}\in\Phi^{k_\max-1}(x))$, which implies:
	\begin{equation}\label{yo2}
		\begin{aligned}
			\min\limits_{x\in\reg}\pr(\tau(x)=k_\max)=
			\min\limits_{x\in\reg}\pr(\tilde{\zeta}_{k_\max-1,x}\in\Phi^{k_\max-1}(x))
		\end{aligned}
	\end{equation}
	From \eqref{yo} and \eqref{yo2}, observe that it suffices to focus on quantities $\pr(\tilde{\zeta}_{s,x}\in\Phi^{s}(x))$ for any $s<k_\max$, in order to determine $\min_{x\in\reg}\pr(\tau(x)=s)$ or $\min_{x\in\reg}\pr(\tau(x)=k_\max)$. The proposition below paves the way for computing such probabilities:
	\begin{proposition}\label{prop_probability_zeta_tilde_phi}
		For all $s\in\{1,2,\dots,k_\max\}$ and $x\in\real^n$:
		\begin{equation*}
			\pr(\tilde{\zeta}_{s,x}\in\Phi^{s}(x)) = \int_{\Phi^{s}(0)}\normal(z|\mu_1(s,x),\Sigma_{\tilde{\zeta}_{s,x}})dz,
		\end{equation*}
		where:
		\begin{equation*}
			\mu_1(s,x) = \begin{bmatrix}
				e^{A}(I+A^{-1}BK)-A^{-1}BK-I_n\\ e^{2A}(I+A^{-1}BK)-A^{-1}BK-I_n\\ \vdots \\e^{sA}(I+A^{-1}BK)-A^{-1}BK-I_n
			\end{bmatrix}x
		\end{equation*}
	\end{proposition}
	\begin{proof}
		See Appendix.
	\end{proof}
	Using the above proposition, we arrive at the following:
	\begin{corollary}[to Proposition \ref{prop_probability_zeta_tilde_phi}]
		The following holds:
		\begin{equation}\label{minimization_event_times}
			\begin{aligned}
				\min\limits_{x\in\reg}\pr(\tilde{\zeta}_{s,x}\in\Phi^{s}(x))=
				\min\limits_{y\in\mu_1(s,\reg)}\int_{\Phi^{s}(0)}\normal(z|y,\Sigma_{\tilde{\zeta}_{s,x}})dz&
			\end{aligned}
		\end{equation}
		where $\mu_1(s,\reg )$ is the convex polytope: $\mu_1(s,\reg)=\{y\in\real^{sn}:y=\mu_1(s,x),x\in\reg\}$. The same holds for $\max_{x\in\reg}$.
	\end{corollary}
	\begin{proof}
		It is a straightforward result of Proposition \ref{prop_probability_zeta_tilde_phi}. Note that $\mu_1(s,\reg)$ is a convex polytope in $\real^{sn}$, since $\mu_1(s,x)$ is linear on $x$ and $\reg$ is a convex polytope.
	\end{proof}
	Optimization problem \eqref{minimization_event_times} requires optimizing the integral of a Gaussian over a polytope (that is, $\Phi(0)$), with its mean varying in a convex polytope (that is, $\mu_1(s,\reg)$). Similar problems have been effectively solved in \cite{luca_tac_2021}: for obtaining the optimal points $y_{opt}\in\mu_1(s,\reg)$, certain KKT-like conditions have been derived, facilitating optimization. To compute the optimal value of the integral at $y_{opt}$, we could: either 1) apply a whitening transformation $T$ transforming $\Sigma_{\tilde{\zeta}_{s,x}}$ to identity (such a $T$ always exists), under/over-approximate $T\cdot\Phi^{s}(0)$ by hyperrectangles, and compute the integral by breaking it down to uni-dimensional integrals (as done in \cite{luca_tac_2021}), or 2) use simple numerical techniques to obtain bounds on it.
	\begin{remark}
		A different way of solving optimization problem \eqref{minimization_event_times}, is by regular convex-optimization methods. By modifying the proof of {\cite[Proposition 2]{Luca_hscc_2019}}, we can prove that the integral \eqref{minimization_event_times} is log-concave on $y$ (omitted, due to space limitations).
	\end{remark}
	\begin{remark}\label{remark_luca_extension_affinity}
		Compared to \cite{luca_tac_2021}, where $\min\limits_{x\in\reg}\pr(\tilde{\zeta}_{s,x}\in S)$ was determined for any fixed polytope $S$, here an added difficulty is that we have to compute $\min\limits_{x\in\reg}\pr(\tilde{\zeta}_{s,x}\in \Phi^s(x))$; that is, a probability of landing in a set varying w.r.t. the optimization variable $x$. This issue has been circumvented, by employing the affinity of $\Phi(x)$ w.r.t. $x$ (see proof of Proposition \ref{prop_probability_zeta_tilde_phi}).
	\end{remark}
	Once we have computed $\max_{x\in\reg}\pr(\tilde{\zeta}_{s,x}\in\Phi^{s}(x))$ and $\min_{x\in\reg}\pr(\tilde{\zeta}_{s-1,x}\in\Phi^{s-1}(x))$ as described above, we readily obtain values for $\min_{x\in\reg}\pr(\tau(x)=s)$, with $s<k_\max$, and $\min_{x\in\reg}\pr(\tau(x)=k_\max)$, by \eqref{yo} and \eqref{yo2} respectively. Similar steps are followed for maxima.
	
	\subsection{Conditional Probabilities on States at Event Times}
	To complete the computation of the transition probability intervals, what is left is determining the first term of \eqref{IV_eq1}: $\min\limits_{x\in\reg}\pr(\zeta(s;x)\in S|\tau(x)=s)$. For any set $Z\subseteq\real^n$ and any $l\in\mathbb{N}_{[0,s]}$ , we adopt the shorthand notation:
	\begin{align*}
		&\pr(Z|\Phi^{l}(x)) \equiv \pr(\zeta_{s,x}\in Z|\tilde{\zeta}_{l,x}\in\Phi^{l}(x))\\
		&\pr(Z|\Phi^{l-1}(x)\times\overline{\Phi}(x)) \equiv\\&\qquad\qquad\qquad \pr(\zeta_{s,x}\in Z|\tilde{\zeta}_{l,x}\in\Phi^{l-1}(x)\times\overline{\Phi}(x))
	\end{align*}
	
	\begin{proposition}\label{prop_min_conditional}
		For any $s<k_\max$, the following holds:
		\begin{equation}\label{min_conditional}
			\begin{aligned}
				&\min_{x\in\reg}\pr(\zeta_{s,x}\in S|\tau(x)=s)\geq \\ &\frac{1}{(1-\min_{x\in\reg}\pr(\Phi(x)|\Phi^{s-1}(x)))}\cdot\bigg[\min_{x\in\reg}\pr(S|\Phi^{s-1}(x)) -\\&\qquad\qquad\qquad\max_{x\in\reg}\pr(S|\Phi^{s}(x))\cdot\max_{x\in\reg}\pr(\Phi(x)|\Phi^{s-1}(x))\bigg]
			\end{aligned}
		\end{equation}
		For $s=k_\max$, the following holds:
		\begin{equation}\label{min_conditional2}
			\begin{aligned}
				&\min_{x\in\reg}\pr(\zeta_{s,x}\in S|\tau(x)=k_\max)=\min_{x\in\reg}\pr( S|\Phi^{k_\max-1})
			\end{aligned}
		\end{equation}
		Similar results hold for $\max_{x\in\reg}\pr(\zeta_{s,x}\in S|\tau(x)=s)$.
	\end{proposition}
	\begin{proof}See Appendix.
	\end{proof}
	Equations \eqref{min_conditional} and \eqref{min_conditional2} indicate that it suffices to focus on quantities $\pr(S|\Phi^l(x))$ and $\pr(\Phi(x)|\Phi^l(x))$, with $l<k_\max$.
	The following result provides the probability distribution giving rise to these conditional probabilities:
	\begin{corollary}[to Proposition \ref{normal}]\label{cor_conditional_normal}
		Consider the random variable $\xi=(\zeta_{s,x}|\tilde{\zeta}_{l,x}=v)$, where $l\in\mathbb{N}_{[0,s]}$, and $v\in\real^{ln}$. Then $\xi\sim\normal(\mu_\xi(x,v), \Sigma_\xi)$, where:
		\begin{align*}
			&\mu_\xi(x,v) = \expec(\zeta_{s,x})-\Sigma_{\zeta_{s,x},\tilde{\zeta}_{l,x}}\Sigma_{\tilde{\zeta}_{l,x}}^{-1}(v-\expec(\tilde{\zeta}_{l,x}))\\
			&\Sigma_\xi = \Sigma_{\zeta_{s,x}} - \Sigma_{\zeta_{s,x},\tilde{\zeta}_{l,x}}\Sigma_{\tilde{\zeta}_{l,x}}^{-1}\Sigma_{\tilde{\zeta}_{l,x},\zeta_{s,x}},
		\end{align*}
		where $\Sigma_{\zeta_{s,x}} = \cov(s,s)$, $\Sigma_{\tilde{\zeta}_{l,x}}$, $\expec(\tilde{\zeta}_{l,x})$ and $\expec(\zeta_{s,x})$ obtained from Proposition \ref{normal}, and $\Sigma_{\zeta_{s,x},\tilde{\zeta}_{l,x}}=\Sigma_{\tilde{\zeta}_{l,x},\zeta_{s,x}}^\top = \begin{bmatrix}
			\cov(s,1) &\cov(s,2) &\dots &\cov(s,l)
		\end{bmatrix}$.
	\end{corollary}
	\begin{proof}
		Straightforward application of the well-known formula for conditional normal distributions.
	\end{proof}
	Based on the above proposition, we get the following:
	\begin{proposition}\label{prop_minimization_conditional}
		For any $l\in\mathbb{N}_{[0,s]}$:
		\begin{equation}\label{minimization_conditional}
			\begin{aligned}
				&\min_{x\in\reg}\pr(S|\Phi^{l}(x))\geq
				\min_{y\in\mathcal{P}_1}\int_{S}\normal(z|y,\Sigma_\xi)dz\\
				&\min_{x\in\reg}\pr(\Phi(x)|\Phi^{l}(x))\geq
				\min_{y\in\mathcal{P}_2}\int_{\Phi(0)}\normal(z|y,\Sigma_\xi)dz
			\end{aligned}
		\end{equation}
		where $\mathcal{P}_1$ and $\mathcal{P}_2$ are the polytopes:
		\begin{align*}
			&\mathcal{P}_1=\{y\in\real^{n}:y=\mu_\xi(x,v+\{x\}^l), x\in \reg, v\in\Phi^l(0)\}\\
			&\mathcal{P}_2=\{y\in\real^{n}:y=\mu_\xi(x,v+\{x\}^l)\hspace{-.5mm}-\hspace{-.5mm}x, x\in \reg, v\in\Phi^l(0)\}
		\end{align*}
		The same holds for $\max_{x\in\reg}$.
	\end{proposition}
	\begin{proof}
		See Appendix
	\end{proof}
	All individual terms in \eqref{min_conditional} and \eqref{min_conditional2} can be obtained via solving optimization problems like \eqref{minimization_conditional}, which can be solved similarly to \eqref{minimization_event_times}. As soon as we determine all such terms, we obtain a lower bound on $\min_{x\in\reg}\pr(\zeta_{s,x}\in S|\tau(x)=s)$. 
	Finally, now that both quantities on the right hand-side of \eqref{IV_eq1} have been computed, we obtain $\check{P}$ and $\hat{P}$ as described right below \eqref{IV_eq1}, and the construction of $\sys_\imc$ is complete. 
	\section{Conclusion and Future Work}
	We have abstracted the sampling behaviour of stochastic linear PETC systems via IMCs. Specifically, we have constructed IMCs and corresponding suitable rewards, that can be employed for computing bounds on expectations of functions of sequences of state-measurements and interevent times, which can be expressed as discounted cumulative sums of rewards. We have demonstrated that such functions can express various sampling performance indicators, and argued that our results are extendable to more general functions, encoding total or average rewards, $\omega-$regular properties, etc. Unfortunately, numerical examples have been omitted, due to space limitations. Future work will focus on the following: a) considering more general functions and sampling indicators, b) endowing the IMCs with actions, enabling scheduling of ETC traffic in networks with probabilistic safety guarantees, c) extending to more general classes of systems and triggering functions, and d) providing extensive experimental results.
	
	\section*{Appendix}
	\begin{proof}[\textbf{Proof of Proposition \ref{prop_probability_zeta_tilde_phi}}]
		$\Phi(x)$ is affine on $x$: $\Phi(x)=\Phi(0)+\{x\}$, where `+' here denotes Minkowski sum. Employing Proposition \ref{normal}:
		\begin{align*}
			\pr(\tilde{\zeta}_{s,x}\in\Phi^{s}(x)) = \int_{\Phi^{s}(x)}\normal(z|\mu_{\tilde{\zeta}_{s,x}},\Sigma_{\tilde{\zeta}_{s,x}})dz&\\
			=\int_{\Phi^{s}(0)+\{x\}^{s}}\normal(z|\mu_{\tilde{\zeta}_{s,x}},\Sigma_{\tilde{\zeta}_{s,x}})dz&\\
			=\int_{\Phi^{s}(0)}\normal(z|\underbrace{\mu_{\tilde{\zeta}_{s,x}}-\{x\}^{s}}_{\mu_1(s,x)},\Sigma_{\tilde{\zeta}_{s,x}})dz&
		\end{align*}
	\end{proof}
	\begin{proof}[\textbf{Proof of Proposition \ref{prop_min_conditional}}]
		From the law of total probability:
		\begin{align*}
			&\pr(S|\Phi^{s-1}(x)) =\\ &=\pr(S|\Phi^{s-1}(x)\times\overline{\Phi}(x))\cdot\pr(\overline{\Phi}(x)|\Phi^{s-1}(x)) +\\ &\qquad\qquad\qquad\pr(S|\Phi^{s}(x))\cdot\pr(\Phi(x)|\Phi^{s-1}(x))\\
			&=\pr(S|\Phi^{s-1}(x)\times\overline{\Phi}(x))\cdot(1-\pr(\Phi(x)|\Phi^{s-1}(x))) +\\ &\qquad\qquad\qquad\pr(S|\Phi^{s}(x))\cdot\pr(\Phi(x)|\Phi^{s-1}(x))
		\end{align*}
		Thus:
		\begin{align*}
			&\pr(\zeta_{s,x}\in S|\tau(x)=s)=\\
			&\pr(S|\Phi^{s-1}(x)\times\overline{\Phi}(x)) = \frac{1}{(1-\pr(\Phi(x)|\Phi^{s-1}(x)))}\cdot\\&\cdot\bigg[\pr(S|\Phi^{s-1}(x)) -\pr(S|\Phi^{s}(x))\cdot\pr(\Phi(x)|\Phi^{s-1}(x))\bigg]
		\end{align*}
		It is then clear how \eqref{min_conditional} is derived. Deriving \eqref{min_conditional2} is straightforward.
	\end{proof}
	\begin{proof}[\textbf{Proof of Proposition \ref{prop_minimization_conditional}}]
		By Corollary \ref{cor_conditional_normal}, we have:
		\begin{equation*}
			\pr(\zeta_{s,x}\in S|\tilde{\zeta}_{l,x}=v)=\int_{S}\normal(z|\mu_\xi(x,v),\Sigma_\xi)dz
		\end{equation*}
		Thus:
		\begin{align*}
			\min_{x\in\reg}\pr(S|\Phi^l(x)) =\min_{x\in\reg}\pr(\zeta_{s,x}\in S|\tilde{\zeta}_{l,x}\in\Phi^l(x))\geq&\\
			\min_{(x,v)\in\reg\times\Phi^l(x)}\int_{S}\normal(z|\mu_\xi(x,v),\Sigma_\xi)dz=&\\
			\min_{(x,v)\in\reg\times\Phi^l(0)}\int_{S}\normal(z|\mu_\xi(x,v+\{x\}^l),\Sigma_\xi)dz=&\\
			\min_{y\in\mathcal{P}_1}\int_{S}\normal(z|y,\Sigma_\xi)dz
		\end{align*} 
		where we made the change of variables: $y=\mu_\xi(x,v+\{x\}^l)$. Similarly:
		\begin{align*}
			\min_{x\in\reg}\pr(\Phi(x)|\Phi^l(x)) \geq&\\
			\min_{(x,v)\in\reg\times\Phi^l(0)}\int_{\Phi(x)}\normal(z|\mu_\xi(x,v+\{x\}^l),\Sigma_\xi)dz=&\\
			\min_{(x,v)\in\reg\times\Phi^l(0)}\int_{\Phi(0)}\normal(z|\mu_\xi(x,v+\{x\}^l)-x,\Sigma_\xi)dz=&\\
			\min_{y\in\mathcal{P}_2}\int_{\Phi(0)}\normal(z|y,\Sigma_\xi)dz&
		\end{align*}

		Since $\mu_\xi(x,v-\{x\}^l)$ is linear on both $x$ and $v$ (see Corollary \ref{cor_conditional_normal}), then $\mathcal{P}_1,\mathcal{P}_2$ are polytopes.
	\end{proof}

	\bibliography{stoch_abs_bib.bib}
	\bibliographystyle{IEEEtran}

\end{document}